\documentclass[12pt]{amsart}
\usepackage{amsmath,amssymb,amsfonts,amsthm,amsopn}
\usepackage{graphicx,tikz}
\usepackage[all]{xy}
\usepackage{amsmath}
\usepackage{multirow, longtable, makecell, caption, array,enumitem}
\usepackage{amssymb}
\usepackage{mathtools}
\usepackage{bookmark}
\usepackage{hyperref}
\mathtoolsset{showonlyrefs}

\makeatletter
\newcommand{\pushright}[1]{\ifmeasuring@#1\else\omit\hfill$\displaystyle#1$\fi\ignorespaces}
\newcommand{\pushleft}[1]{\ifmeasuring@#1\else\omit$\displaystyle#1$\hfill\fi\ignorespaces}
\makeatother
\calclayout

\newtheorem{theorem}{Theorem}[section]
\newtheorem{lemma}[theorem]{Lemma}

\newtheorem{proposition}[theorem]{Proposition}
\newtheorem{definition}[theorem]{Definition}

\newtheorem{remark}[theorem]{Remark}

\theoremstyle{definition}

\usepackage{cellspace} %
\newcommand\numberthis{\addtocounter{equation}{1}\tag{\theequation}}

\def\bR{\mathbb{R}}
\def\bC{\mathbb{C}}
\def\bN{\mathbb{N}}

\def\cF{\mathcal{F}}
\def\cS{\mathcal{S}}

\def\cM{\mathcal{M}}

\def\rd{\bR^d}

\def\rdd{\bR^{2d}}

\def\la{\langle}
\def\ra{\rangle}
\def\lc{\left(}
\def\rc{\right)}

\def\wt{\widetilde}
\def\*b{*_{\bullet}}
\def\w{\mathrm{w}}

\def\op{\mathrm{op}}

\def\Bd'{B_{\delta'}}

\def\cBd'{\bar{B}_{\delta'}}

\def\Sp{\mathrm{Sp}(d,\bR)}
\def\Mp{\mathrm{Mp}(d,\bR)}

\def\muS{\mu(S)}

\newcommand{\GLL}{\mathrm{GL}(2d,\bR)}
\def\smo{\setminus\{0\}}

\def\mrd{\bR^{d\times d}}

\def\Sp{\mathrm{Sp}(d,\bR)}

\def\Mp{\mathrm{Mp}(d,\bR)}

\def\ird{\int_{\rd}}
\def\irdd{\int_{\rdd}}
\def\wh{\widehat}

\def\cU{\mathcal{U}}

\def\Sjo{M^{\infty,1}}

\newcommand{\hbi}{\frac{i}{\hbar}}

\def\te{\mathfrak{E}}
\def\rte{\bR \setminus \te}
\def\loc{\mathrm{loc}}

\def\tauhz0{\widehat{\mathcal{T}}^\hbar(z_0)}
\def\tauhz{\widehat{\mathcal{T}}^\hbar(z)}
\begin{document}
	
\title[On a novel time slicing approximation for Feynman path integrals]{On the convergence of a novel family of time slicing approximation operators for Feynman path integrals}
\author{S. Ivan Trapasso}
\address{MaLGa Center - Department of Mathematics, University of Genoa, via Dodecaneso 35, 16146 Genova, Italy}
\email{salvatoreivan.trapasso@unige.it}
%\date{}
%\date{\today} 

\begin{abstract}
In this note we study the properties of a sequence of approximate propagators for the Schr\"odinger equation, in the spirit of Feynman's path integrals. Precisely, we consider Hamiltonian operators arising as the Weyl quantization of a quadratic form in phase space, plus a bounded potential perturbation in the form of a pseudodifferential operator with a rough symbol. It is known that the corresponding Schr\"odinger propagator is a generalized metaplectic operator. This naturally motivates the introduction of a manageable time slicing approximation consisting of operators of the same type. By means of techniques and function spaces of time-frequency analysis it is possible to obtain several convergence results with precise rates in terms of the mesh size of the time slicing subdivision. In particular, we prove convergence in the norm operator topology in $L^2$, as well as pointwise convergence of the corresponding integral kernels for non-exceptional times.  
\end{abstract}
\maketitle
\section{Introduction}\label{sec intro} The rigorous analysis of Feynman path integrals has been, and still is, a challenging source of intriguing problems for mathematicians from manifold areas. The knowledge accumulated over the last seventy years encompasses several aspects, ranging from foundational issues to applied ones; the interested reader can consult the monographs \cite{AlbeHKM_book_08,Fujiwara_book_17,Mazzucchi_book_09} for in-depth studies on the topic. 

From the point of view of operator theory, the \textit{sequential approach} to path integrals essentially deals with providing explicit representation formulas for the Schr\"odinger evolution operator $U(t,s)$ in terms of sequences of approximate propagators on $L^2(\rd)$; the latter are required to converge to $U(t,s)$ in some sense, usually with respect to the strong operator topology. A prime example, inspired by the custom in physics, is to use approximate propagators arising from the Trotter formula \cite{Nelson_64}, namely
\[ 	U(t,s)f =\lim_{n\to \infty} E_n(t,s) f, \quad f \in L^2(\rd), \qquad E_n(t,s) \coloneqq  \left(e^{-\hbi   \frac{t-s}{n}H_{0}}e^{-\hbi \frac{t-s}{n}V}\right)^{n}. \] 
Different choices are of course possible; in fact, a careful design of the approximate propagators is crucial in order to obtain stronger convergence results. For instance, the time slicing approximations introduced by Fujiwara in \cite{Fujiwara_fundsol_79,Fujiwara_duke_80} rely on delicate tools of oscillatory integral operators and semiclassical analysis; the need of sophisticated techniques to handle them is counterweighted by deep convergence results in the norm operator topology and also at the subtler level of kernels, even in the semiclassical limit $\hbar \to 0$. 

\subsection{The role of Gabor analysis} Only in recent times the techniques of Gabor wave packet analysis have been successfully brought into play in the study of mathematical path integrals, leading to relevant advances and new directions to be explored. 

By way of illustration let us mention here the papers \cite{N_convLp_16} and \cite{NT_CMP20} - an expository overview of these and other related results can be found in \cite{T_pathsurvey}. The power of a phase-space approach is fully displayed in the first contribution, where it allows one to ``break the barriers'' of the standard $L^2$ setting and derive convergence results for Fujiwara's time slicing approximations in $L^p$-Sobolev spaces with $p \ne 2$. 

The second paper focuses on the pointwise convergence of integral kernels of the Feynman-Trotter approximate propagators $E_n(t,s)$ introduced above. Generally speaking, this issue has been heuristically conjectured by Feynman himself \cite{Feynman_48,FeynmanHibbs_10} and remained a widely open problem for a long time, at least until the pioneering results by Fujiwara already mentioned above, for sufficiently regular potentials and more refined parametrices. In the paper \cite{NT_CMP20} the problem has been solved for quadratic Hamiltonian and bounded potential perturbations with low regularity, mainly using tools and function spaces of Gabor analysis \cite{CR_book,Gro_book}.

This is a convenient stage where to briefly discuss the basic notions of time-frequency analysis, both for clarity and future reference. Recall that a phase-space representation of a signal $u$ can be obtained as a decomposition with respect to \textit{Gabor wave packets} of the form \[ \pi(z)g(y) = e^{2\pi i \xi \cdot y}g(y-x), \quad z=(x,\xi) \in \rdd, \] where $g$ is a fixed (non-trivial) function on $\rd$ that is well localized in the time-frequency space $\rd \times \wh{\rd} \simeq \rdd$. To be more precise, the function $V_g u (z) \coloneqq \la u, \pi(z)g \ra$ of $z$ resulting from duality pairing between $u\in \cS'(\rd)$ and $\pi(z)g$ with $g \in \cS(\rd)$ is called the \textit{Gabor transform} of $u$. Lifting a problem to the phase space level is usually convenient in many respects; for instance, the Gabor transform of any temperate distribution happens to be a continuous function with at most polynomial growth. 

A concrete, alternative point of view on the procedure just described comes from signal analysis: if $u \in L^2(\rd)$, for instance, then we have explicitly
\[ V_g u (x,\xi) = \ird e^{-2\pi i \xi \cdot y} u(y)\overline{g(y-x)} \, dy = \cF (u\cdot \overline{g(\cdot - x)})(\xi), \] and thus the Gabor transform corresponds to taking the Fourier transform of the variable slice of the signal $u$ obtained by localization near $x$ with a sliding window function $g$; this explains why the Gabor transform is also widely known as the \textit{short-time Fourier transform}. 

\textit{Modulation spaces} can be naturally introduced at this stage as spaces of distributions characterized by prescribed summability conditions for the corresponding phase-space representations \cite{BenyiOko_book}. For $1 \le p \le \infty$ and a fixed $g \in \cS(\rd)\smo$ we set
\[ M^p(\rd) \coloneqq \{ u \in \cS'(\rd) : \| u \|_{M^p} < \infty \}, \quad \|u \|_{M^p} \coloneqq \|V_g u \|_{L^p(\rdd)}. \]
Modulation spaces constitute a family of Banach spaces, increasing in $p$, whose norm is stable under change of window (namely, different choices of $g$ result in equivalent norms). We emphasize that many typical function spaces of harmonic analysis turn out to coincide with (generalized) modulation spaces. The most important example is $M^2(\rd) \simeq L^2(\rd)$; together with the extremal spaces $M^1(\rd)$ (the original Feichtinger algebra \cite{Fei_segal81,Jakobsen_segal18}) and $M^\infty(\rd)$ (the space of \textit{mild distributions} \cite{Fei_mild19}) they provide a convenient framework for basic Fourier analysis - namely, a Banach-Gelfand triple \cite{FLC_triples08}, also known in quantum physics and spectral theory as the formalism of rigged Hilbert spaces. 

\subsection{The Schr\"odinger equation} 
Modulation spaces are widely used in the aforementioned papers \cite{N_convLp_16,NT_CMP20}  to rigorously frame the problem of path integrals from a phase space perspective, but also to tune the regularity of potential perturbations, following an established series of results for the Schr\"odinger equation - see the monograph \cite{CR_book} for a state-of-the-art account. 

In this note we consider the Cauchy problem for the Schr\"odinger equation 
\begin{equation}\label{eq schro base}
\begin{cases}
i \partial_t \psi = 2\pi(H_0 + V)\psi \\ 
\psi(s,x)= f(x)
\end{cases}
\end{equation}
with initial datum $f\in \cS(\rd)$ at time $t=s$. Note that we fixed $\hbar=1/2\pi$ here - some comments on this choice are given below. The Hamiltonian operator consists of two parts. First, we assume $H_0 = Q^{\w}$, namely $H_0$ is the Weyl quantization of a homogeneous quadratic polynomial $Q:\rdd \to \bR$. We also consider a potential perturbation of the form $V=\sigma(t,\cdot)^{\w}$, where the symbol $\sigma(t,\cdot)=\sigma_t(\cdot)$ belongs to the so-called \textit{Sj\"ostrand class} $\Sjo(\rdd)$ \cite{Sjo_94}, that is a modulation space characterized by a mixed Lebesgue-type norm in phase-space: 
\[ \| \sigma_t \|_{\Sjo} \coloneqq \irdd \sup_{z \in \rdd} |V_\Phi \sigma_t (z,\zeta)| \, d\zeta < \infty, \quad \Phi \in \cS(\rdd)\smo. \]
We additionally assume that the correspondence $\bR \ni t\mapsto \sigma(t,\cdot) \in M^{\infty,1}(\rdd)$ is continuous in a mild sense - see Definition \ref{def nar conv} below for further details. 

For concreteness, let us discuss here the case where $V$ is a standard multiplication operator by a function in the Sj\"ostrand class (see Section \ref{sec weyl}). To give a flavour of typical potentials in this family we remark that the latter includes bounded functions that locally coincide with the Fourier transform of an integrable function, hence also continuous (cf.\ Section \ref{sec sjo}). Actually, this is the least regularity condition on potentials, since members of $\Sjo(\rd)$ are not differentiable in general; nevertheless, a relevant subclass of $\Sjo(\rd)$ is the space $C^\infty_{\mathrm{b}}(\rd)$ of smooth bounded functions with bounded derivatives of any order. We also have $\cF \cM(\rd) \subset \Sjo(\rd)$, where $\cF\cM(\rd)$ is the space of Fourier transforms of complex finite measures (cf.\ \cite[Proposition 3.4]{NT_CMP20}). Potentials of this type already appeared several times in the literature on mathematical path integrals, especially in connection with the Parseval duality approach introduced by It\^o \cite{Ito_61} and developed by Albeverio et al.\ \cite{Albeverio_inventiones_77,AlbeHKM_book_08}. \medskip

Let us consider first the Hamiltonian $H_0=Q^\w$ as above. This setting covers some fundamental examples like the free particle and the harmonic oscillator, possibly in the presence of a uniform magnetic field. A classical result of harmonic analysis states that the evolution operator $U_0(t,s) = e^{-2\pi i (t-s) H_0}$ is a \textit{metaplectic operator} associated with the phase space flow $\bR \ni \tau \mapsto S_{\tau}\in \Sp$ of the corresponding classical system, that is  $U_0(t,s)=c(t-s)\mu(S_{t-s})$ for some $c=c(t-s) \in \bC$ with $|c|=1$ - see \cite{dG_sympmeth11,Folland_book} and Section \ref{sec met} for an expanded account. \smallskip

The structure of the full Schr\"odinger propagator $U(t,s)$ in the presence of a potential $V$ as above can be investigated by means of standard arguments from the theory of operator semigroups that have their roots in the \textit{perturbation method} of quantum mechanics. Let us briefly recall the main steps for the sake of clarity - more details and references are given in Section \ref{sec gen met} below. The problem \eqref{eq schro base} can be recast in integral form in accordance with Duhamel's principle, namely \begin{equation}
\psi(t,x) = U_0(t,s)f(x) -2\pi i \int_s^t U_0(t,\tau)\sigma_\tau^{\w}\psi(\tau,x)d\tau. 
\end{equation} 
After setting $\varphi(t,x)\coloneqq U_0(s,t)\psi(t,x)$ (which leads to the the so-called interaction picture) and using the evolution property of $U_0$ we have
\begin{equation}\label{eq schro volt}
\varphi(t,x) = f(x) - 2\pi i \int_s^t U_0(s,\tau)\sigma_\tau^{\w}U_0(\tau,s)\varphi(\tau,x)d\tau.
\end{equation}
The symplectic covariance of Weyl calculus (see \eqref{eq symp cov weyl} below) yields
\begin{equation}\label{}
U_0(s,\tau)\sigma_\tau^{\w}U_0(\tau,s) = (\sigma_\tau \circ S_{\tau-s})^{\w} \eqqcolon b(\tau,s)^{\w}.
\end{equation}
The solution of the Volterra integral equation \eqref{eq schro volt} is thus given by 
\begin{equation}\label{}
\varphi(t,x) = a^{\w}(t,s)f(x), 
\end{equation}
where the symbol $a(t,s)$ can be expressed as a Dyson-Phillips expansion:
\begin{equation}\label{eq def ats}
a(t,s) \coloneqq 1 + \sum_{n\ge 1} \lc -2\pi i\rc^n \int_s^t \int_s^{t_1} \ldots \int_s^{t_{n-1}} b(t_1,s)\#\cdots \# b(t_n,s) \, dt_n \ldots dt_1,
\end{equation} and $\#$ denotes Weyl's twisted product of symbols (see Section \ref{sec weyl} for further details). It is possible to show that $a(t,s)\in M^{\infty,1}(\rdd)$ (cf.\ Lemma \ref{lem stima norma a} below). Therefore, the original problem is solved by
\begin{equation}\label{eq propag}
\psi(t,x) = U(t,s)f(x), \quad U(t,s) = U_0(t,s)a(t,s)^{\w}. 
\end{equation}
It is intuitively clear that $U(t,s)$ is intimately connected to the homogenous propagator $U_0(t,s)$ described before. The underlying relationship has been completely elucidated in the papers \cite{CGNR_wienerFIO_13,CGNR_genmet14}, leading to the introduction of the class $FIO(S_{t-s})$ of \textit{generalized metaplectic operators} associated with the flow $S_{t-s}$ - more details are collected in Section \ref{sec gen met} for convenience. 

\subsection{Discussion of the main results} Let us now introduce a novel time slicing approximation in the spirit of Feynman's formulation. Inspired by \eqref{eq def ats} and \eqref{eq propag}, for $0<t-s\le T$ it is natural to consider a short-time $FIO$-type approximate propagator $E(t,s)$ such as
\begin{equation}\label{eq def Ets}
	E(t,s) \coloneqq U_0(t,s)e(t,s)^{\w}, 
\end{equation} where the symbol $e(t,s)$ is defined as follows:
\begin{equation}\label{eq def ets}
	e(t,s) \coloneqq \exp \lc -2\pi i \int_s^t b(\tau,s)d\tau \rc. 
\end{equation} We prove below that $e(t,s) \in \Sjo(\rdd)$ and also that it is a good short-time approximation of $a(t,s)$ in \eqref{eq def ats}, precisely $\|a(t,s)-e(t,s)\|_{\Sjo} = O((t-s)^2)$ - cf.\ Lemma \ref{lem est diff a-e} below.

Consider now a subdivision $\Omega = \{t_0,\ldots,t_L\}$ of the interval $[s,t]$ such that $s=t_0 < t_1 < \ldots < t_L = t$. We accordingly define the time slicing approximate propagator as
\begin{equation}\label{def Eots}
	E(\Omega;t,s) \coloneqq E(t_L,t_{L-1})\cdots E(t_1,t_0).
\end{equation} 
It turns out that generalized metaplectic operators are well behaved under composition (see Theorem \ref{thm gen met properties} below), hence $E(\Omega;t,s)$ is again an operator in the same class and there exists a symbol $e(\Omega;t,s) \in \Sjo(\rdd)$ such that $E(\Omega;t,s) = U_0(t,s)e(\Omega;t,s)^{\w}$. To be precise, iterated application of the symplectic covariance property of Weyl calculus yield
\begin{equation}\label{eq eots}
	e(\Omega;t,s)=\wt{e}(t_L,t_{L-1}) \# \cdots \# \wt{e}(t_1,t_0),
\end{equation} where the modified symbols $\wt{e}(t_{j+1},t_j)$, $j=0,\ldots,L-1$, are defined by
\begin{equation}\label{eq etilde}
	\wt{e}(t_{j+1},t_{j}) \coloneqq e(t_{j+1},t_{j})\circ S_{t_{j}-t_0}, \quad j=0,\ldots, L-1.
\end{equation}

We stress that the operators involved in this time slicing approximation are different from those considered by Fujiwara as well as from those associated with the Trotter formula. Nevertheless, the short-time propagator $E(t,s)$ is quite easy to handle thanks to the natural algebraic properties of the $FIO$ family and the symbol $e(\Omega;t,s)$ could be in principle represented as an integral - cf.\ e.g.\ \cite{Wong_weylbook98} for explicit formulas. 

It should be emphasized that, under the same assumptions for \eqref{eq schro base} outlined above, the Feynman-Trotter parametrices happen to be generalized metaplectic operators, as shown in \cite{NT_CMP20}, but they are worse than the novel ones introduced here as far as the short-time approximation power is concerned. This is indeed the core of our first convergence result, that can be stated as follows. 
\begin{theorem}\label{mainthm}
	Fix $T>0$ and let $s,t \in \bR$ be such that $0<t-s\le T$. Consider a subdivision $\Omega = \{t_0,\ldots,t_L\}$ of the interval $[s,t]$ such that $s=t_0 < t_1 < \ldots < t_L = t$. Consider $a(t,s)$ and $e(\Omega;t,s)$ as defined in \eqref{eq def ats} and \eqref{eq eots} respectively. Then $a(t,s), e(\Omega;t,s) \in \Sjo(\rdd)$ and there exists $C=C(T)>0$ such that
	\begin{equation}\label{eq conv main}
		\| e(\Omega;t,s) - a(t,s)\|_{\Sjo} \le C \omega(\Omega). 
	\end{equation}
As a result, $e(\Omega;t,s) \to a(t,s)$ uniformly in $\rdd$ if $\omega(\Omega)\to 0$ and 
\begin{equation}\label{eq conv op norm} \|E(\Omega;t,s) - U(t,s) \|_{L^2\to L^2} \le  C \omega(\Omega). \end{equation}
\end{theorem}
In short, our approach is based on a generalized version of the one pioneered in the aforementioned works by Fujiwara. While it is customary to use arguments of this type for deriving (norm) convergence results for operators, such as \eqref{eq conv op norm}, in view of the structure of generalized metaplectic operators here we focus first on the level of symbols. The effort is repaid by the estimate \eqref{eq conv main}, which also provides a precise rate of convergence with respect to the mesh size of the subdivision, uniformly with respect to $t-s$. This is precisely the case where the advantages of a time slicing approximation that is designed to best fit the features of the problem reflect into stronger results. We emphasize that this issue is particularly relevant in connection with the Trotter product formula: in spite of its widespread use in the path integral literature, it just provides a (qualitative) strong convergence result that can be hardly improved in the unitary setting (i.e., convergence in norm or even at the level of kernels) - see e.g.\ the discussion in \cite[Appendix D]{zagreb}.

We already mentioned the problem of pointwise convergence of integral kernels in time slicing approximations. It is quite natural to consider this problem in the framework offered by generalized metaplectic operators: as a matter of fact, the latter do have a typical explicit form as integral operators - see Section \ref{sec gen met} for further details. Let $\te$ the set of \textit{exceptional times} for $FIO(S_{t-s})$, namely if $\displaystyle S_{t-s} = \begin{bmatrix} A_{t-s} & B_{t-s} \\ C_{t-s} & D_{t-s} \end{bmatrix}$ is the block decomposition of the classical flow, then $\te = \{ \tau \in \bR : \det B_{\tau} =0 \}$. It can be proved that if $U(t,s) \in FIO(S_{t-s})$ and $t-s \in \rte$ then there exist a phase factor $c=c(t-s) \in \bC$, $|c|=1$, and an amplitude function $a'(t,s) = a'(t,s,\cdot) \in \Sjo(\rdd)$ such that
\begin{equation}\label{eq gen met int rep}
	U(t,s)f(x) = \ird u(t,s,x,y)f(y) \, dy,
\end{equation} 
\[  u(t,s,x,y) \coloneqq c(t-s) \lvert\det B_{t-s}\rvert^{-1/2} e^{2\pi i \Phi_{t-s}(x,y)} a'(t,s,x,y), \] where  $\Phi_{t-s}(x,y)$ is a quadratic polynomial whose coefficients depend only on the entries of $S_{t-s}$. Under the same assumptions one can prove that there exists $e'(\Omega;t,s) \in \Sjo(\rdd)$ such that 
\begin{equation}\label{eq Eots int rep}
	E(\Omega;t,s)f(x) = \ird k(\Omega;t,s,x,y) f(y) \, dy,  
\end{equation}
\begin{equation}\label{eq Eots ker}
	k(\Omega;t,s,x,y) \coloneqq c(t-s) \lvert\det B_{t-s}\rvert^{-1/2} e^{2\pi i \Phi_{t-s}(x,y)} e'(\Omega;t,s,x,y).
\end{equation}
The estimates in Theorem \ref{mainthm} can thus be used to infer refined convergence results (precisely, in $(\cF L^1)_{\mathrm{loc}}(\rdd) = (\Sjo)_{\mathrm{loc}}(\rdd)$) at the level of integral kernels. 
\begin{theorem}\label{thm int ker}
	For any $t-s \in (0,T] \setminus \te$, consider the explicit representations of $U(t,s)$ and $E(\Omega;t,s)$ in \eqref{eq gen met int rep} and \eqref{eq Eots int rep} respectively. Under the same assumptions of Theorem \ref{mainthm}, we have:
	\begin{enumerate}
		\item The amplitudes $a'(t,s,\cdot)$ and $e'(\Omega;t,s,\cdot)$ belong to $\Sjo(\rdd)$ and there exists $C_1=C_1(t-s)$ such that 		
		\begin{equation}\label{eq int ampl bound}
			\|e'(\Omega;t,s,\cdot) - a'(t,s,\cdot) \|_{\Sjo} \le C_1 \omega(\Omega).
		\end{equation}
	\item For any real-valued function $\Psi$ on $\rdd$ with compact support there exists $C_2=C_2(t-s,\Psi)>0$ 
	\begin{equation}\label{eq int ker cov bound}
		\| [k(\Omega;t,s,\cdot) - u(t,s,\cdot)] \Psi \|_{\cF L^1} \le C_2 \omega(\Omega).
	\end{equation}
	As a result, we have that $k(\Omega;t,s,\cdot) \to u(t,s,\cdot)$ locally uniformly in $\rdd$. 
	\end{enumerate}
\end{theorem}
This result should be compared with those proved in \cite{NT_CMP20} for the Feynman-Trotter parametrices. They are similar in nature, also for what concerns the proof strategy; a key difference is that here we are able to control the rate of convergence in \eqref{eq int ker cov bound}, thanks to quantitative estimates on convergence of the amplitudes such as \eqref{eq conv main}.  We also remark that it should not be surprising that the previous bounds are not uniform with respect to $t-s$, in contrast to those in Theorem \ref{mainthm}, in view of the possibly degenerate behaviour of the oscillatory integral phase $\Phi_{t-s}$ and ``normalization'' $\lvert\det B_{t-s}\rvert^{-1/2}$ at exceptional times - see the proof of Theorem \ref{thm gen met ker} below for more details. 

As a final remark, we emphasize that the choice $\hbar= 1/2\pi$ is not restrictive as soon as $\hbar$ is interpreted as a small fixed parameter. The more challenging problem of the semiclassical limit, namely convergence also for $\hbar \to 0$, would require the occurrence of \textit{positive} powers of $\hbar$ in bounds such as \eqref{eq conv main}. While this issue would certainly fall within the scope of our analysis, a nice behaviour with respect to $\hbar$ is not to be expected in the context of perturbation theory, which basically involve expansions in \textit{negative} powers of $\hbar$. In view of the lack of semiclassical approximation power for our parametrices, we also preferred to abstain from a careful revision of the basic notions of Gabor analysis in order to keep track of $\hbar$ in a consistent way, see e.g. \cite{CDGN_semi}.

To conclude, we are aware that assuming boundedness of potential perturbations is a quite restrictive condition for the purposes of physics. In fact, this is a problem shared by several mathematical frameworks for the analysis of path integrals; however, thanks to Gabor analysis we were able to significantly relax the regularity assumptions for the potentials. In future work we expect to be able to relax the boundedness assumption too and cover rough potential perturbations with quadratic growth, after a substantial rearrangement of the framework outlined in this contribution. Nevertheless, there is reason to believe that techniques and function spaces of Gabor analysis can be successfully employed to tackle this and other problems of mathematical physics. 

\section{Preliminary results}
\subsection{Notation} We set $|t|^2=t\cdot t$, $t\in\rd$, where $x\cdot y$ is the scalar product on $\rd$. The bracket  $\la  f,g\ra$ denotes the extension to $\cS' (\rd)\times \cS (\rd)$ of the inner product $\la f,g\ra =\int_{\rd} f(t){\overline {g(t)}}dt$ on $L^2(\rd)$. 

We choose the following normalization for the Fourier transform:
\[ \cF(f)(\xi) = \hat{f}(\xi) = \int_{\rd} e^{-2\pi i x\cdot \xi} f(x)dx, \quad \xi \in \rd.
\] 

\subsection{Weyl pseudodifferential operators}\label{sec weyl} 
The usual definition of the Weyl operator $\sigma^\w = \op_{\w}(\sigma)$ with symbol $\sigma:\mathbb{R}^{2d}\rightarrow\mathbb{C}$
is
\begin{equation}\label{eq def wig}
	\sigma^\w f(x)\coloneqq \int_{\mathbb{R}^{2d}}e^{2\pi i(x-y)\cdot \xi}\sigma\left(\frac{x+y}{2},\xi\right)f(y) \, dyd\xi.
\end{equation}
The way to rigorously interpret this (formal) integral operator relies on the function spaces to which the symbol $\sigma$ and the function $f$ belong. For instance, classical symbol classes such as H\"ormander's ones $S_{\rho,\delta}^{m}(\mathbb{R}^{2d})$ \cite{HormanderIII} are usually defined by means of decay/smoothness conditions.

One can also approach the issue from the point of view of time-frequency analysis \cite{Gro_ped}: a straightforward computation (yet formal in general) shows that 
\begin{equation}\label{eq weyl wig}
\la \sigma^\w f, g \ra = \la \sigma, W(g,f) \ra,	
\end{equation} where we have introduced the (cross-)\textit{Wigner transform} of $f,g$:
\begin{equation}\label{eq intro wig}
	W(f,g)(x,\xi) \coloneqq  \ird e^{-2\pi i \xi \cdot y} f \lc x+\frac{y}{2}\rc  \overline{g\lc x-\frac{y}{2}\rc} dy.
\end{equation} 
This is in fact a standard representation in the phase-space formulation of quantum mechanics \cite{QMPS_book}. It turns out that there is an intimate connection between the Gabor transform and the Wigner distribution; for instance, the $L^p$ norm of the (cross-)Wigner transform can be equivalently used to measure the modulation space regularity of a signal \cite{dG_sympmeth11}. 

We are thus lead to interpret \eqref{eq def wig} in the weak sense and consider \eqref{eq weyl wig} as a definition of $\sigma^\w$ for a generalized symbol $\sigma \in \cS'(\rdd)$, for any $f,g \in \cS(\rd)$. Note that this choice paves the way to using modulation spaces both as symbol classes as well as target spaces for Weyl operators; the interested reader can find more details on this perspective in \cite{Gro_book, CR_book}. 

By way of illustration let us mention that the multiplication by a function $V(x)$ is a special example of Weyl operator with symbol \begin{equation} \sigma_V (x,\xi)=(V\otimes 1)(x,\xi)= V(x), \quad (x,\xi)\in \rdd.\end{equation} 
It is not difficult to prove that the correspondence $V\mapsto \sigma_V$ is continuous from $M^{\infty,1}(\rd)$ to $M^{\infty,1}(\rdd)$ - see the next section for other properties of the Sj\"ostrand class. One may similarly prove that a Fourier multiplier with symbol $m(\xi)$ is a Weyl operator with symbol $(1\otimes m)(x,\xi) = m(\xi)$.

The symbolic calculus relies on the composition of Weyl transforms, which provides a bilinear form on symbols known as the \textit{Weyl} (or \textit{twisted}) \textit{product}:
\begin{equation} \sigma^{\mathrm{w}} \circ \rho^{\mathrm{w}} = (\sigma \# \rho)^{\mathrm{w}}, \quad \sigma \# \rho = \cF^{-1}(\hat{\sigma} \natural \hat{\rho}), \end{equation}
where the \textit{twisted convolution} \cite{Gro_book} of $\hat{\sigma}$ and $\hat{\rho}$ is (formally) defined by
\begin{equation}\label{eq def twist conv}
	(\hat{\sigma} \natural \hat{\rho})(x,\xi) \coloneqq  \irdd e^{\pi i (x,\xi) \cdot (\eta,-y)} \hat{\sigma}(y,\eta) \hat{\rho}(x-y,\xi-\eta)\, dy d\eta.
\end{equation} We remark that in phase space quantum mechanics and deformation quantization it is customary to refer to $\sigma\#\rho$ as the \textit{Moyal star product} of $\sigma$ and $\rho$ after \cite{Moyal_49}.

% cf.\ e.g.\ \cite{Estrada:89,Voros_78,QMPS_book}. 

%In any case, this is a fundamental notion in order to establish a non-commutative algebra structure on symbol spaces; as far as modulation spaces are concerned, the problem has been studied in several papers (cf.\ for instance \cite{CToftWahl_weylprod_14,Gro_weylprod_06,HolstToftWahl_weylprod_07}). 

\subsection{Some results on the Sj\"ostrand class} \label{sec sjo}
We already introduced the Sj\"ostrand class in Section \ref{sec intro}. As the name suggests, it was first presented by Sj\"ostrand in \cite{Sjo_94} as an exotic class of non-smooth symbols still giving bounded pseudodifferential operators in $L^2(\rd)$. It was later rediscovered in Gabor analysis by Gr\"ochenig \cite{Gro_sjo,GroRze}, who obtained novel proofs of known results but also a number of new important characterizations (e.g., almost diagonalization in phase space of the corresponding Weyl operators and the Wiener property). 

We recall that, as a fully fledged modulation space, $\Sjo(\rd)$ can be designed the collection of $u \in \cS'(\rd)$ such that, for some (hence any) non-trivial $g \in \cS(\rd)$,
\[ \| u \|_{\Sjo} \coloneqq \ird \sup_{x \in \rdd} |V_g u (x,\xi)| \, d\xi < \infty. \]
We list here some properties that will be used below. The reader can consult the aforementioned papers for a more comprehensive account. 
\begin{proposition}\label{prop sjo prop} 
	\begin{enumerate}[label=(\roman*)]
		\item $M^{\infty,1}(\rd) \subset (\cF L^1(\rd))_{\loc}\cap L^{\infty}(\rd) \subset C_{\mathrm{b}} (\rd)$, where $C_{\mathrm{b}} (\rd)$ is the space of bounded and continuous functions $\rd \to \bC$.  
		\item $(M^{\infty,1}(\rd))_{\loc} = (\cF L^1(\rd))_{\loc}$.
		\item $\cF \cM(\rd) \subset M^{\infty,1}(\rd)$, where $\cM(\rd)$ is the space of complex finite measures on $\rd$.
		\item $\Sjo(\rd)$ is a Banach algebra under pointwise multiplication:
		\[ f,g \in \Sjo(\rd) \, \Rightarrow \, fg \in \Sjo(\rd). \]
		\item $\Sjo(\rd)$ is a Banach algebra under the Weyl product of symbols:
		\[ \rho,\sigma \in \Sjo(\rd) \, \Rightarrow \, \rho \# \sigma \in \Sjo(\rd). \] 
		\item If $\sigma \in \Sjo(\rdd)$ then $\sigma^\w$ is a bounded operator on $L^2(\rd)$ and, in general, on $M^p(\rd)$ for $1\le p \le \infty$. To be precise, there exists $C>0$ independent of $\sigma$ such that \begin{equation}\label{eq cont sjo}\| \sigma^\w\|_{M^p\to M^p} \le C \|\sigma\|_{\Sjo}. \end{equation}
	\end{enumerate}
\end{proposition}
\begin{proof}
	\begin{enumerate}[label=(\roman*)]
		\item It is a direct consequence of the definition. 
		\item See \cite[Proposition 2.9]{BenyiOko_book} for a proof.
		\item A proof can be found in \cite[Proposition 3.4]{NT_CMP20}.
		\item This is a special case of a more general characterization, see \cite[Theorem 3.5 and Corollary 2.10]{ReichSickel_algmod16}.
		\item Proofs can be found in the original paper \cite{Sjo_94} by Sj\"ostrand as well as in the already mentioned paper \cite{Gro_sjo} by Gr\"ochenig. 
		\item The same comments of the previous item apply here. For a streamlined textbook proof see also \cite[Theorem 14.5.2]{Gro_book}. \qedhere
	\end{enumerate}
\end{proof}

\begin{remark}\label{rem algebra}
%	To unambiguously fix the notation: whenever concerned with a product of elements $a_1,\ldots,a_N$ in a Banach algebra $(A,\star)$, we write \begin{equation}
%		\prod_{k=1}^N a_k \coloneqq  a_1 \star a_2 \star \ldots \star a_N.
%	\end{equation}
%	This relation is meant to hold even when $(A,\star)$ is a non-commutative algebra, provided that the symbol on the left-hand side exactly designates the ordered product on the right-hand side.
	
Let $(A,\star)$ be a unital Banach algebra; in particular, there exists $C>0$ such that 
	\[ \|a_1 \star a_2\|_A \le C \|a_1 \|_A \|a_2\|_A, \quad a_1,a_2 \in A. \] Recall that it is always possible to introduce an equivalent norm on $A$ for which the product estimate above holds with $C=1$ and the unit has norm equal to $1$ (cf.\ e.g. \cite[Theorem 10.2]{RudinFA}). From now on, we assume to work with such equivalent norm whenever concerned with a Banach algebra. 
\end{remark}

It was already noted in \cite{Sjo_94} that $C^{\infty}_{\mathrm{c}}(\rd)$ is not dense in $M^{\infty,1}(\rd)$ with the norm topology, as a consequence of the particular strength of the latter. It is thus natural to weaken the notion of convergence and introduce the following notion \cite{CNR_rough15,Sjo_94}.

\begin{definition}\label{def nar conv} 
	Let $\Omega$ be a subset of some Euclidean space. The map $\Omega \ni \nu \mapsto \sigma_\nu \in M^{\infty,1}(\rd)$ is said to be continuous for the narrow convergence if: 
	\begin{enumerate}[label=(\roman*)]
		\item it is a continuous map in $\cS'(\rd)$ (weakly), and
		\item there exists a function $h \in L^1(\rd)$ such that for some (hence any) window $g\in \cS(\rd)$ one has $\sup_{x\in\rd}| V_g\sigma_\nu(x,\xi)| \le h(\xi)$ for any $\nu \in \Omega$ and a.e. $\xi \in \rd$.
	\end{enumerate}
\end{definition}
It turns out that $\cS(\rd)$ is dense in $M^{\infty,1}(\rd)$ with respect to the narrow convergence \cite{Sjo_94}. The notion of narrow convergence happens to be useful for our purposes as it implies the uniform bound $\sup_{\nu \in \Omega} \| \sigma_\nu\|_{\Sjo} \le \|h \|_{L^1} < \infty$.

Finally, we recall from \cite[Lemma 2.2]{CGNR_genmet14} and \cite[Lemma 3.1]{NT_CMP20} a result on uniform estimates for linear changes of variable in the Sj\"ostrand class. In fact, here we consider a slightly more general version.
\begin{lemma}\label{lem unif stima norma}
	Consider $\sigma\in  M^{\infty,1}(\rdd)$
	and a correspondence $\bR \ni t\mapsto M_t\in \GLL$. For any $t \in \bR$ we have $\sigma \circ M_t\in \Sjo(\rdd)$; precisely, for any fixed $\Phi \in \cS(\rdd) \smo$,
	\begin{equation}\label{eq est linear comp prec}
		\| \sigma\circ M_t \|_{\Sjo}\le C \lvert \det M_t \rvert^{-1} \|V_{\Phi \circ M_t^{-1}} \Phi \|_{L^1} \|\sigma\|_{\Sjo}, 
	\end{equation} for some constant $C=C(\Phi)>0$. 
	As a result, if $t\mapsto M_t$ is continuous, for any $T>0$ there exists $C(T)>0$ such that
	\begin{equation}\label{eq est linear comp}
		\| \sigma\circ M_t \|_{\Sjo}\le C(T)  \|\sigma\|_{\Sjo}, \quad t \in [-T,T].
	\end{equation}
\end{lemma}

\subsection{Metaplectic operators and quadratic Hamiltonian}\label{sec met}
There are several equivalent ways to define metaplectic operators. The reader is referred to the monographs \cite{dG_sympmeth11,Folland_book,Taylor_nch} for comprehensive treatments of the topic. 

Here we confine ourselves to recall that the metaplectic group coincides with the two-fold covering of the symplectic group $\Sp$. There exists a faithful and strongly continuous unitary representation of $\Mp$ on $L^2(\rd)$, called the Shale-Weil \textit{metaplectic representation}, which allows us to recast the issue in terms of a correspondence between any symplectic matrix $S\in\Sp$ and a pair of unitary \textit{metaplectic operators} differing only by the sign, both denoted by $\mu(S)$ with a slight abuse of notation.

Among the several properties satisfied by metaplectic operators, it is important to our purposes to highlight that the Weyl quantization satisfies a characterizing intertwining relationship called \textit{symplectic covariance} (cf.\ \cite[Theorem 215]{dG_sympmeth11}): for every symbol $\sigma \in \cS'(\rdd)$ we have
\begin{equation}\label{eq symp cov weyl}
			(\sigma \circ S)^\w = \mu(S)^{-1} \sigma^\w \mu(S). 
\end{equation}
Moreover, metaplectic operators can be well characterized in terms of their action on phase space thanks to the analysis of the corresponding \textit{Gabor matrix} $G_{S}(z,w)\coloneqq \la \mu(S)\pi(z)g,\pi(w)\gamma \ra$, $z,w \in \rdd$, for fixed $g,\gamma \in \cS(\rd)$. In particular, one can prove rapid decay estimates for $|G_S(z,w)|$ away from the graph of $S$ in $\rdd$ - this can be viewed in a sense as a manifestation of the correspondence principle of quantum mechanics. The reader is invited to consult the monograph \cite{CR_book} for an in-depth study of the issue. 

While boundedness results on modulation spaces for metaplectic operators can be readily derived from phase-space estimates such as those for the corresponding Gabor matrix, we need below a more precise bound where the dependence on $S$ of the underlying constants is completely clarified. A result of this type was proved in \cite[Corollary 3.4]{CNT_dispdiag20} and reads as follows: for every $S \in \Sp$ and $1\le p \le \infty$ there exists an absolute constant $C>0$ such that, for every $f \in M^p(\rd)$,
\begin{equation}\label{eq bound met op}
	\| \mu(S)f \|_{M^p} \le C (\sigma_1(S) \cdots \sigma_d(S))^{|1/2-1/p|} \|f \|_{M^p},
\end{equation} where $\sigma_1(S) \ge \ldots \ge \sigma_d(S) \ge 1$ are the $d$ largest singular values of $S$. 

A concrete characterization of metaplectic operators can be given in terms of Schr\"odinger propagators with quadratic Hamiltonian \cite{RobbinSalamon}. Consider the Cauchy problem for the Schr\"odinger equation 
\begin{equation}\label{}
	\begin{cases} i \partial_t \psi = 2\pi H_0 \psi \\ \psi(s,x) = f(x),
	\end{cases}
\end{equation} 
where $H_{0}=Q^{\w}$ is the Weyl quantization a real-valued, time-independent, quadratic homogeneous polynomial $Q$ on $\rdd$. Precisely, if 
\begin{equation} Q (x,\xi) = \frac{1}{2} A \xi \cdot \xi + Bx \cdot \xi + \frac{1}{2} Cx \cdot x,
\end{equation} for some matrices $A,B,C \in \mrd$ with $A=A^\top$ and $C=C^\top$, then
\[ H_0 = Q^\w = -\frac{1}{8\pi^2} \sum_{j,k=1}^d A_{j,k}\partial_j\partial_k -\frac{i}{2\pi}\sum_{j,k=1}^d B_{j,k}x_j\partial_k  + \frac{1}{2}\sum_{j,k=1}^d C_{j,k}x_j x_k- \frac{i}{4\pi}\mathrm{Tr}(B). \]

%It is known that $H_{0}= Q^{\w}$ is a self-adjoint operator on the maximal domain (see \cite{Hormander_mehler95})
%\[ D\left(H_{0}\right)=\{ f \in L^{2}(\mathbb{R}^{d})\,:\,H_{0} f \in L^{2}(\mathbb{R}^{d})\} .\] 

\noindent It is well known that the associated propagator is a metaplectic operator, precisely 
\begin{equation}
	U_0(t,s) = e^{-2\pi i (t-s) Q^\w} = c(t-s) \mu(S_{t-s}), 
\end{equation} for some $ c=c(t-s)\in \bC$ with $|c|=1$, where the (continuous) mapping 
\begin{equation}\label{eq block St}
	\bR \ni \tau \mapsto S_{\tau} = \begin{bmatrix} A_{\tau} & B_{\tau} \\ C_{\tau} & D_{\tau} \end{bmatrix} \in \Sp
\end{equation} is the phase-space flow determined by the Hamilton equations for the corresponding classical model with Hamiltonian $Q(x,\xi)$.

Another classical result concerns the explicit characterization of the propagator $U_0(t,s)$ as an integral operator, cf.\ e.g.\ \cite[Theorems 4.51 and 4.53]{Folland_book} and \cite{dG_sympmeth11}. Recall that the matrix $S_{\tau}$ is said to be \textit{free} if the upper-right block $B_{\tau}$ in \eqref{eq block St} is invertible. In this case, we have
		\begin{equation} \label{met int formula} U_0(t,s)f(x)= c(t-s) \lvert\det B_{t-s}\rvert^{-1/2} \int_{\rd} e^{2\pi i \Phi_{t-s}(x,y)} f(y) dy, \qquad f \in \cS(\rd), \end{equation} again for some $c=c(t-s) \in \bC$ with $|c|=1$, where we introduced the quadratic form (also known as the \textit{generating function} of $S_{t-s}$)
		\begin{equation}\label{eq phit}
			\Phi_{t-s} (x,y) = \frac{1}{2}D_{t-s} B_{t-s}^{-1}x\cdot x - B_{t-s}^{-1} x \cdot y + \frac{1}{2} B_{t-s}^{-1}A_{t-s} y\cdot y.
		\end{equation} 
	
For future reference we define the set of \textit{exceptional times} as
		\begin{equation}\label{eq def te} \mathfrak{E} = \{ \tau \in \bR \,:\, \det B_\tau =0 \},  \end{equation}
		namely as the collection of values of $\tau$ such that $S_\tau$ is \textit{not} a free symplectic matrix. Some of the properties of this set can be immediately inferred from the fact that it actually coincides with the zero set of an analytic function: apart from the case $\mathfrak{E}=\bR$ (which trivially occurs when $H_0=0$), $\mathfrak{E}$ is a discrete (hence at most countable) subset of $\bR$ which always includes $\tau=0$ and possibly only this value (this is the case of the free particle).
\subsection{Generalized metaplectic operators}\label{sec gen met}
Let us consider now the perturbed problem in \eqref{eq schro base} with the assumptions stated in Section \ref{sec intro}. The arguments of perturbation theory already anticipated there can be used to give a rigorous proof of the following facts (cf.\ \cite[Theorem 4.1]{CGNR_genmet14} for the details): the problem under consideration is globally backward and forward well-posed in $L^2(\rd)$ and the corresponding full Schr\"odinger propagator $U(t,s)$ is a one-parameter strongly continuous group of automorphisms of $L^2(\rd)$ - in fact, of any modulation space $M^p(\rd)$, $1 \le p \le \infty$, hence the phase-space concentration is preserved under the evolution. The structure of $U(t,s)$ can be obtained as a perturbation of the propagator of the associated homogeneous problem $U_0(t,s) = c(t-s)\mu(S_{t-s})$, precisely 
\[ U(t,s) = U_0(t,s)a(t,s)^\w, \] where the symbol $a(t,s) \in \Sjo(\rdd)$ (cf. Lemma \ref{lem stima norma a} below) is the one defined in \eqref{eq def ats}. 

Let us emphasize that that $U(t,s)$ is not a unitary propagator in general, unless $V$ is self-adjoint - equivalently, if $\sigma_t$ is a real-valued symbol. 

Operators of the form of $U(t,s)$ as above, namely arising as combinations of metaplectic and pseudodifferential operators with symbols in $\Sjo(\rdd)$, have been thoroughly studied in time-frequency analysis: they constitute the family $FIO(S_{t-s})$ of \textit{generalized metaplectic operators} introduced in \cite{CGNR_wienerFIO_13,CGNR_genmet14}. Roughly speaking, this can be thought of as the largest class of ``perturbations'' of $\mu(S_{t-s})$ that still evolve Gabor wave packets in phase space essentially following the classical flow $S_{t-s}$ - a precise condition can be given in terms of the Gabor matrix decay. It is clear that if metaplectic operators are used as building blocks for a rigorous theory of path integrals (as proposed for instance in \cite{dG_principles17,RobbinSalamon}), generalized metaplectic operators designed as before naturally provide a rigorous framework to handle perturbation expansions such as the Born approximation (cf.\ e.g.\ \cite[Chapter 6]{FeynmanHibbs_10}). 

In view of this discussion, the following properties of generalized metaplectic operators should not be surprising. 

\begin{theorem}\label{thm gen met properties} Let $S,S_1,S_2 \in \Sp$.
	\begin{enumerate}[label=(\roman*)]
	\item Let $T:\cS(\rd)\to \cS'(\rd)$ be a linear continuous operator. $T \in FIO(S)$ if and only if there exist $\sigma_1,\sigma_2 \in M^{\infty,1}(\rdd)$ such that 
		\begin{equation} T = \sigma_1^{\w} \muS = \muS  \sigma_2^{\w}. \end{equation} In particular, $\sigma_2 = \sigma_1 \circ S$. 
	\item An operator $T \in FIO(S)$ is bounded on $M^p(\rd)$ for any $1\le p \le \infty$.
	\item If $T_1 \in FIO(S_1)$ and $T_2 \in FIO(S_2)$, then $T_1T_2 \in FIO(S_1S_2)$.
	\item If $T \in FIO(S)$ is invertible on $L^2(\rd)$ then $T^{-1} \in FIO(S^{-1})$. 
	\end{enumerate}
\end{theorem}

The following result extends \eqref{met int formula} and thus provides a concrete picture of generalized metaplectic operators for non-exceptional times. 
\begin{theorem}\label{thm gen met ker} Consider the problem \eqref{eq schro base} and let $U(t,s)=U_0(t,s)a(t,s)^\w \in FIO(S_{t-s})$ the corresponding propagator as discussed above - cf.\ in particular \eqref{eq gen met int rep}. For $t-s \in (0,T] \setminus \te$, there exists a symbol $a'(t,s) = a'(t,s,\cdot) \in \Sjo(\rdd)$ such that 
	\begin{equation}\label{}
		U(t,s)f(x) = c(t-s) \lvert\det B_{t-s}\rvert^{-1/2} \ird e^{2\pi i \Phi_{t-s}(x,y)} a'(t,s,x,y)f(y) dy.
	\end{equation}
Moreover, there exists $C=C(t-s)$ such that 
\begin{equation}\label{eq est a'}
 \| a'(t,s) \|_{\Sjo} \le C \|	a(t,s) \|_{\Sjo}.
\end{equation}
\end{theorem}
\begin{proof}
	The representation of $T \in FIO(S)$ as a Fourier integral operator is proved in \cite[Theorem 5.1]{CGNR_genmet14}, so that the claim follows after easy modifications - cf. e.g.\ \cite[Lemma 3.1]{NT_CMP20}. It remains to prove the estimate \eqref{eq est a'}. In the aforementioned results (see also \cite[Proposition 5.2]{CGNR_genmet14}) it is shown by direct computation that
	\[
	a'(t,s) = \cU_2 \cU \cU_1 (a(t,s)\circ S_{t-s}^{-1}),
	\]
	where $\cU,\cU_1,\cU_2$ are the mappings defined on $\sigma \in \Sjo(\rdd)$ by
	\[\cU_1\sigma =\sigma\circ U_1, \quad \cU_2\sigma=\sigma\circ U_2,\quad \widehat{\cU\sigma}(\xi,\eta)=e^{\pi i \xi \cdot \eta}\widehat{\sigma}(\xi,\eta),\]
	\[ U_1 = \begin{bmatrix}
		I & O \\ A_{t-s} & I
	\end{bmatrix}, \quad U_2 = \begin{bmatrix}
	I & O \\ O & B_{t-s}^\top
\end{bmatrix}, \] where $A_{t-s}$ and $B_{t-s}$ come from the block decomposition of $S_{t-s}$ in \eqref{eq block St}. For what concerns $\cU$, the proof of \cite[Corollary 14.5.5]{Gro_book} precisely shows that $\Sjo(\rdd)$ is invariant under the action of $\cU$. The desired conclusion thus follows by repeated application of Lemma \ref{lem unif stima norma} - note that $U_1$ is a symplectic matrix for any $t-s$, while $\det U_2 = \det B_{t-s}$; as a result, \eqref{eq est linear comp prec} shall be used and uniformity with respect to $t-s$ cannot be achieved. 
\end{proof}
\begin{remark} The previous representation result extends to any $T \in FIO(S)$ - in particular it holds for $E(\Omega;t,s) \in FIO(S_{t-s})$ in \eqref{def Eots}, as claimed in \eqref{eq Eots int rep}. We emphasize that also in this case we have that 
	\begin{equation}\label{eq e' est}
		\| e'(\Omega;t,s) \|_{\Sjo} \le C \|e(\Omega;t,s)\|_{\Sjo},
	\end{equation} for a constant $C=C(t-s)>0$ that depends on $t-s \in (0,T] \notin \te$. 
\end{remark}
\section{Proof of the main results}
\subsection{Preliminary estimates}
For this and the subsequent sections we refer to the problem \eqref{eq schro base} with the stated assumptions. We also use the notation introduced in Section \ref{sec intro}.

We first prove that the symbols $a(t,s)$ belong to a bounded subset of $\Sjo(\rdd)$. 
\begin{lemma}\label{lem stima norma a}
	Fix $T>0$ and let $s,t \in \bR$ be such that $0<t-s\le T$. Let $a(t,s)$ be defined as in \eqref{eq def ats}. Then $a(t,s) \in \Sjo(\rdd)$ and there exists $C=C(T)>0$ such that
	\begin{equation}\label{eq stima norma a}
		\| a(t,s) \|_{\Sjo} \le C. 
	\end{equation}
\end{lemma}
\begin{proof}
Let us set 
	\begin{equation}\label{eq def alphan} \alpha_n(t,s) \coloneqq \int_s^t \int_s^{t_1} \ldots \int_s^{t_{n-1}} b(t_1,s)\#\cdots \# b(t_n,s) \, dt_n \ldots dt_1, \quad n \in \bN, \end{equation}
	so that we can write $a(t,s) = 1+\sum_{n\ge 1} \lc -2\pi i \rc^n \alpha_n(t,s)$. It is clear that $\alpha_n(t,s) \in \Sjo(\rdd)$ since $(\Sjo,\#)$ is a Banach algebra (cf.\ Proposition \ref{prop sjo prop}) and $b(\tau,s) = \sigma_\tau \circ S_{\tau-s} \in \Sjo(\rdd)$ for all $\tau>s$ in view of Lemma \ref{lem unif stima norma}. To be precise, we have
	\begin{align*}
		\| \alpha_n(t,s) \|_{\Sjo}  & = \left\| \int_s^t \int_s^{t_1} \ldots \int_s^{t_{n-1}} b(t_1,s)\#\cdots \# b(t_n,s) \, dt_n \ldots dt_1 \right\|_{\Sjo} \\
		& \le \int_s^t \int_s^{t_1} \ldots \int_s^{t_{n-1}} \| b(t_1,s)\|_{\Sjo} \cdots \| b(t_n,s)\|_{\Sjo} \, dt_n \ldots dt_1 \\
		& \le \frac{(t-s)^n}{n!} \lc \sup_{\tau \in [s,t]} \|b(\tau,s)\|_{\Sjo} \rc^n  \\ 
		& \le \frac{C_0(T)^n(t-s)^n}{n!} \lc \sup_{\tau \in [s,t]} \|\sigma_\tau \|_{\Sjo} \rc^n \\ 
		& \le \frac{C_1(T)^n}{n!}(t-s)^n, \numberthis \label{eq ats norm}
	\end{align*} where we set $C_1(T)=C_0(T)\lc \sup_{\tau \in [s,t]} \|\sigma_\tau\|_{\Sjo} \rc$, in particular $C_0(T)$ comes from the estimate \eqref{eq est linear comp} and the supremum is finite since $t\mapsto \sigma_t$ is continuous for the narrow convergence. The claim follows from \eqref{eq ats norm}, since
	\[ \|a(t,s)\|_{\Sjo} \le 1 + \sum_{n\ge 1} (2\pi)^n \|\alpha_n(t,s)\|_{\Sjo} \le e^{C_1(T)T/2\pi} \eqqcolon C. \qedhere \] 
\end{proof}

We prove below that the symbols $e(t,s)$ introduced in \eqref{eq def ets} are good short-time approximations of $a(t,s)$. 
\begin{lemma}\label{lem est diff a-e}
	Fix $T>0$ and let $s,t \in \bR$ be such that $0<t-s\le T$. Consider $a(t,s)$ and $e(t,s)$ as defined in \eqref{eq def ats} and \eqref{eq def ets} respectively. We have $a(t,s),e(t,s) \in \Sjo(\rdd)$ and there exists $C=C(T)>0$ such that
	\begin{equation}\label{eq est diff a-e}
		\| a(t,s) - e(t,s) \|_{\Sjo} \le C (t-s)^2.
	\end{equation}
\end{lemma}
\begin{proof} First of all, note that 
	\[ e(t,s) = \exp\lc -2\pi i \int_s^t b(\tau,s)d\tau\rc = \sum_{n\ge 0} \lc -2\pi i \rc^n \frac{1}{n!} \lc  \int_s^t b(\tau,s)d\tau\rc^n.  \] To align the notation with the one introduced before, we define 
	\[ \varepsilon_n(t,s) \coloneqq \frac{1}{n!} \lc  \int_s^t b(\tau,s)d\tau\rc^n, \quad n\in \bN, \] so that $ e(t,s) = 1+\sum_{n\ge 1} \lc -2\pi i \rc^n \varepsilon_n(t,s)$. Recall from Proposition \ref{prop sjo prop} that $\Sjo$ is a Banach algebra under pointwise multiplication, then arguing as in the proof of Lemma \ref{lem stima norma a} we readily obtain
	\[ \|\varepsilon_n(t,s) \|_{\Sjo} \le \frac{C_1(T)^n}{n!} (t-s)^n, \] where the constant $C_1(T)$ equals the one appearing in \eqref{eq ats norm}. 
	 As a result, we have
	\begin{align*}
		\| e(t,s)-a(t,s)\|_{\Sjo} & = \left\| \sum_{n\ge 2} \lc -2\pi i \rc^n \left[ \varepsilon_n(t,s) - \alpha_n(t,s) \right] \right\|_{\Sjo} \\
		& \le \sum_{n\ge 2} (2\pi)^{-n} \lc \| \varepsilon_n(t,s)\|_{\Sjo} + \| \alpha_n(t,s) \|_{\Sjo}\rc \\
		& \le 2\sum_{n\ge 2} \frac{C_1(T)^n}{(2\pi)^n n!}(t-s)^n \\
		& \le C (t-s)^2,
	\end{align*} where we set $ C =(C_1^2/2\pi^2)e^{C_1T/2\pi}$. 
\end{proof}

\subsection{Proof of Theorem \ref{mainthm}}
The strategy of the proof of Theorem \ref{mainthm} is largely inspired by the pioneering one of \cite[Lemma 3.2]{Fujiwara_duke_80}. The latter was subsequently generalized in \cite[Theorem 10]{NT_JMP19}. 
%\begin{proposition}\label{prop main}
%	Let $(B,\star)$ be a Banach algebra. Assume that for some $\delta>0$ there are families $\{ X(t,s),Y(t,s) : t,s \in \bR\} \subset B $ satisfying the following conditions.
%	\begin{enumerate}[label=(\roman*)]
%		\item $Y$ satisfies the composition law $Y(t,\tau)Y(\tau,s)=Y(t,s)$ for every $s<\tau<t$.
%		\item For every $T>0$ there exists a constant $C_0 = C_0(T)\geq 1$ such that 
%		\begin{equation}\label{chiave0}
%			\|Y(t,s)\|_{B}\leq C_0(T) \quad \text{ for } \quad 0<t-s\leq T.
%		\end{equation}
%		\item There exist  $C_1=C_1(T)>0$ and $N \in \bN$ such that
%		\begin{equation}\label{chiave1}
%			\|X(t,s)-Y(t,s)\|_B\leq C_1(T) (t-s)^{N+1} \quad \text{ for }\quad 0<t-s\leq T.
%		\end{equation}
%	\end{enumerate}	Fix $\delta >0$. For any subdivision $\Omega:s=t_0<t_1<\ldots<t_L=t$ of the interval $[s,t]$, with $\omega(\Omega)=\sup\{t_j-t_{j-1}: j=1,\ldots,L\}<\delta$, consider therefore the composition $X(\Omega,t,s)$ defined by 
%	\[
%	X(\Omega,t,s) = X(t_L,t_{L-1}) \star X(t_{L-1},t_{L-2}) \star \cdots \star X(t_1,t_0), 
%	\]
%	Then, for every $T>0$ there exists a constant $C=C(N,T,\delta)>0$ such that 
%	\begin{equation}\label{daverL2}
%		\|X(\Omega,t,s)-Y(t,s)\|_{B}
%		\leq C \omega(\Omega)^{N}(t-s) 
%	\end{equation}for $0<t-s\leq T$. More precisely, $		C=C_0(T)^2C_1(\delta)\exp(C_0(T)C_1(\delta) \delta^N T)$.
%\end{proposition}
\begin{proof}[Proof of Theorem \ref{mainthm}]
Let us focus on deriving \eqref{eq conv main} first. Consider the propagator $U(t,s)$ in \eqref{eq propag}; it is not difficult to realize that the corresponding group property and the symplectic covariance of Weyl calculus (recall that $U_0(t,s)=c(t-s)\mu(S_{t-s})$) imply the following composition law for $s<\tau<t$:
\begin{equation}\label{eq a comp law}
	a(t,s) = (a(t,\tau)\circ S_{\tau-s})\#a(\tau,s).
\end{equation}
Iteration of this procedure referring to the subdivision $\Omega$ in the claim brings us to introduce modified short-time symbols as follows:
\begin{equation}\label{eq atilde}
	\wt{a}(t_{j+1},t_{j}) \coloneqq a(t_{j+1},t_{j})\circ S_{t_{j}-t_0}, \quad j=0,\ldots, L-1.
\end{equation} As a result, we obtain the decomposition
\begin{equation}\label{eq a comp}
	a(t,s) = \wt{a}(t_L,t_{L-1}) \# \ldots \# \wt{a}(t_1,t_0).
\end{equation} 
Note that a combination of Lemmas \ref{lem unif stima norma} and \ref{lem stima norma a} provides the existence of a constant $C_0=C_0(T)\ge 1$ such that
\begin{equation}\label{eq est atilde}
	\| \wt{a}(t_{j+1},t_{j}) \|_{\Sjo} \le C_0(T), \quad j=0,\ldots, L-1.
\end{equation} 
We also remark that the symbols $\wt{a}$ satisfy the composition property
\begin{equation}\label{eq semig atilde}
\wt{a}(t_{j+1},t_j)\#\wt{a}(t_j,t_{j-1}) = \wt{a}(t_{j+1},t_{j-1}) \quad j=1,\ldots,L-1,	
\end{equation} which can be easily verified as follows, using the symplectic covariance property and the composition law \eqref{eq a comp law} (the latter is enough if $j=1$):
\begin{align*}
	\wt{a}(t_{j+1},t_j)^{\w}\wt{a}(t_j,t_{j-1})^{\w} & = U_0(t_0,t_j)a(t_{j+1},t_j)^{\w}U_0(t_j,t_0)U_0(t_0,t_{j-1})a(t_j,t_{j-1})^{\w}U_0(t_{j-1},t_0)
	\\ & = U_0(t_0,t_j)a(t_{j+1},t_j)^{\w}U_0(t_j,t_{j-1})a(t_j,t_{j-1})^{\w}U_0(t_{j-1},t_0) \\
	& =U_0(t_0,t_j)U_0(t_j,t_{j-1}) (a(t_{j+1},t_j)\circ S_{t_j-t_{j-1}})^{\w}a(t_j,t_{j-1})^{\w}U_0(t_{j-1},t_0) \\
	& = U_0(t_0,t_{j-1}) (a(t_{j+1},t_j)\circ S_{t_j-t_{j-1}})^{\w}a(t_j,t_{j-1})^{\w}U_0(t_{j-1},t_0) \\
	& = U_0(t_0,t_{j-1}) a(t_{j+1},t_{j-1})^{\w}U_0(t_{j-1},t_0) \\
	& = (a(t_{j+1},t_{j-1}) \circ S_{t_{j-1}-t_0})^{\w} \\
	& = \wt{a}(t_{j+1},t_{j-1})^{\w}.
\end{align*}

Consider now the approximate propagator $E(\Omega;t,s)$ defined in \eqref{def Eots}; a straightforward application of the symplectic covariance formula shows that
\begin{equation}
	E(\Omega;t,s) =  E(t_L,t_{L-1})\cdots E(t_1,t_0) = U_0(t,s)e(\Omega;t,s)^{\w},
\end{equation} where 
\begin{equation}\label{eq def eots}
	e(\Omega;t,s)=\wt{e}(t_L,t_{L-1}) \# \cdots \# \wt{e}(t_1,t_0),
\end{equation} the modified symbols $\wt{e}(t_{j+1},t_j)$, $j=0,\ldots,L-1$, being defined as in \eqref{eq atilde}. 

Set $r(t,s) \coloneqq \wt{e}(t,s) - \wt{a}(t,s)$; by \eqref{eq est diff a-e} and Lemma \ref{lem unif stima norma} we infer the bound 
\begin{equation}\label{eq bound r}
	\| r(t,s) \|_{\Sjo} \le C_1(T) (t-s)^2, \quad 0 < t-s \le T,
\end{equation} for some $C_1=C_1(T)>0$. Then we can write
\begin{align*} e(\Omega;t,s)-a(t,s) & = \wt{e}(t_L,t_{L-1}) \cdots  \wt{e}(t_1,s) - \wt{a}(t_L,t_{L-1})  \cdots  \wt{a}(t_1,s)\\ &  = (r(t_L,t_{L-1}) + \wt{a}(t_L,t_{L-1}))\cdots(r(t_1,s) + \wt{a}(t_1,s)) \\
	& \quad -\wt{a}(t_L,t_{L-1})  \cdots  \wt{a}(t_1,s) .
\end{align*} From this point forward, the proof is substantially identical to the one given in \cite{NT_JMP19}. We retrace the main steps for the sake of completeness. A careful inspection of the previous expansion reveals that it ultimately consists of the sum of ordered Weyl products of symbols, each of them having the form 
\begin{equation}\label{eq prod type}
\underbrace{\wt{a}\cdots \wt{a}}_{q_{k+1}}\,\underbrace{r\cdots r}_{p_{k}}\, \underbrace{\wt{a}\cdots \wt{a}}_{q_{k}}\,\cdots\, \underbrace{r\cdots r}_{p_{1}}\,\underbrace{\wt{a}\cdots \wt{a}}_{q_{1}},
\end{equation} where $p_1,\ldots,p_k,q_1,\ldots q_k,q_{k+1}$ are non negative integers (in particular $p_j>0$) that sum to $L$. Note also that the symbols $\wt{a}$ in each block of $q_j$ terms can be grouped using \eqref{eq semig atilde}. 

We are now in the position to bound the $\Sjo$ norm of products of the form \eqref{eq prod type} using the previously derived estimates, namely \eqref{eq est atilde} and \eqref{eq bound r}. Precisely - recall that $C_0 \ge 1$, we have 
\begin{align*}
	&\leq C_0^{k+1}\prod_{j=1}^k \prod_{i=1}^{p_j} C_1 (t_{J_j+i}-t_{J_{j}+i-1})^2\\
	&\leq C_0\prod_{j=1}^k \prod_{i=1}^{p_j} C_0C_1 (t_{J_j+i}-t_{J_{j}+i-1})^2,
\end{align*} where we set $J_j=p_1+\ldots+p_{j-1}+q_1+\ldots +q_j$ for $j\geq2$ and $J_1=q_1$. The sum over $p_1,\ldots,p_k,q_1,\ldots,q_{k+1}$ of terms of this type gives in turn 
\begin{align}
	\| e(\Omega;t,s)-a(t,s)\|_{\Sjo}&\leq C_0\left\{ \prod_{j=1}^L(1+C_0C_1 (t_{j}-t_{j-1})^2) -1 \right\}\\
	&\leq C_0\left\{\exp\left( \sum_{j=1}^L C_0C_1 (t_{j}-t_{j-1})^2\right) -1 \right\}\\
	&\leq C_0\left\{\exp\left(C_0C_1\omega(\Omega)(t-s)\right) -1\right\}\\
	&\leq C(T)\omega(\Omega)
\end{align}	where in the last inequality we used $e^{\tau}-1\leq \tau e^\tau$, for $\tau\geq 0$, and we finally set
\[ C=C(T) = C_0^2C_1T\exp\left(C_0C_1T^2\right). \]
This concludes the proof of \eqref{eq conv main} in the claim. \smallskip

For what concerns uniform convergence of $e(\Omega;t,s)$ to $a(t,s)$ in $\rdd$ as $\omega(\Omega)\to 0$, it is an easy consequence of \eqref{eq conv main} in view of the continuous embedding $\Sjo(\rdd)\subset C(\rdd)\cap L^\infty(\rdd)$, cf.\ Proposition \ref{prop sjo prop}. \smallskip

Lastly, let us prove \eqref{eq conv op norm}. Note that
\[ \| E(\Omega;t,s) - U(t,s)\|_{M^p\to M^p} \le \| U_0(t,s) \|_{M^p \to M^p} \|e(\Omega;t,s)^\w - a(t,s)^\w \|_{M^p\to M^p}. \]
The term $\|U_0(t,s)\|_{M^p\to M^p}$ can be bounded by a constant $C_2$, possibly depending on $T$, since the singular values of $S_{t-s}$ in \eqref{eq bound met op} depend continuously on the  entries of $S_{t-s}$ (cf.\ e.g.\ \cite[Corollary A.4.5]{sontag}) and the latter are in turn continuous functions of $t-s \in \bR$ (note that $S_0 = I$). For the remaining term, by \eqref{eq cont sjo} we have
\[ \|e(\Omega;t,s)^\w-a(t,s)^\w \|_{M^p \to M^p} \le C_3 \| e(\Omega;t,s)-a(t,s) \|_{\Sjo}, \] for an absolute constant $C_3$. The bound in \eqref{eq conv op norm} thus holds with $C=C_2(T)C_3$. 
\end{proof}

\subsection{Proof of Theorem \ref{thm int ker}}
Let us finally provide a proof of the convergence result stated in Theorem \ref{thm int ker} at the level of integral kernels. The technique is similar to that used in \cite{NT_CMP20}.
\begin{proof}[Proof of Theorem \ref{thm int ker}] To lighten the notation we set $k(\Omega;t,s,\cdot) = k(\Omega;t,s)$ and $u(t,s,\cdot) = u(t,s)$, and similarly for $e'(\Omega;t,s)$ and $a'(t,s)$.
\begin{enumerate}
	\item We already showed that the amplitudes belong to the Sj\"ostrand class, cf.\ \eqref{eq est a'} and \eqref{eq e' est} respectively. By virtue of Lemma \ref{lem unif stima norma} and \eqref{eq conv main} we infer \eqref{eq int ampl bound} - note from the proof of Theorem \ref{thm gen met ker} that the relationship between symbols of generalized metaplectic operators and amplitudes of the corresponding integral representations only depend on the underlying metaplectic part $U_0(t,s)$. 
	\item Fix a real-valued function $\Psi\in C_{\mathrm{c}}^{\infty}(\mathbb{R}^{2d})$ with compact support, then choose another real-valued function $\Theta\in C_{\mathrm{c}}^{\infty}(\mathbb{R}^{2d})$ with
	$\Theta=1$ on $\mathrm{supp}\Psi$. We have 
	\begin{align}
		\begin{split} & \left\Vert  \mathcal{F}\left[\left(k(\Omega;t,s)-u(t,s)\right) \Psi \right]\right\Vert _{L^1} \\  & \hspace{2.5cm} =  \lvert\det B_{t-s}\rvert^{-1/2}\left\Vert \mathcal{F}\left[e^{2\pi i\Phi_{t-s}}\left(e'(\Omega;t,s)-a'(t,s) \right)\Psi\right]\right\Vert _{L^1} \end{split}\\ 
		&\hspace{2.5cm} =  \lvert\det B_{t-s}\rvert^{-1/2} \left\Vert \mathcal{F}\left[\left(\Theta e^{2\pi i\Phi_{t-s}}\right)\left(e'(\Omega;t,s)-a'(t,s) \right) \Psi \right]\right\Vert _{L^1}\\
		&\hspace{2.5cm} \le \lvert\det B_{t-s}\rvert^{-1/2} \left\Vert \mathcal{F}\left[\Theta e^{2\pi i\Phi_{t-s}}\right]*\mathcal{F}\left[\left(e'(\Omega;t,s)-a'(t,s) \right)\Psi\right]\right\Vert_{L^1}\\
		& \hspace{2.5cm} \le  \lvert\det B_{t-s}\rvert^{-1/2} \left\Vert \mathcal{F}\left[\Theta e^{2\pi i\Phi_{t-s}}\right]\right\Vert _{L^1}\left\Vert \mathcal{F}\left[\left(e'(\Omega;t,s)-a'(t,s) \right) \Psi\right]\right\Vert _{L^1}. 
	\end{align} Since $\Theta e^{2\pi i\Phi_{t-s}}\in C_{\mathrm{c}}^{\infty}(\mathbb{R}^{2d})$, it is clear that we can find a constant $C_0 = C_0(t-s,\Psi)$ such that $\| \cF \left[\Theta e^{2\pi i\Phi_{t-s}}\right]\|_{L^1}< C_0$. Finally, using \eqref{eq conv main},
	\begin{align*} \left\Vert \mathcal{F}\left[\left(e'(\Omega;t,s)-a'(t,s) \right) \Psi\right]\right\Vert _{L^1} & = \| V_\Psi \left(e'(\Omega;t,s)-a'(t,s) \right) (0,\cdot) \|_{L^1} \\ & \le C'(\Psi) \| e'(\Omega;t,s)-a'(t,s) \|_{\Sjo}, \end{align*} and the claimed bound \eqref{eq int ker cov bound} thus follows with $C_2=C_0(t-s)C_1(t-s)C'(\Psi)$, where $C_1(t-s)$ comes from \eqref{eq int ampl bound}. 
	
	Uniform convergence of kernels on compact subsets is then clear: since
	\[ \| \lc k(\Omega;t,s)-u(t,s)\rc \Psi \|_{L^\infty} \le \| \cF \left[ \lc k(\Omega;t,s)-u(t,s) \rc \Psi \right] \|_{L^1},  \qedhere \] 
	for any compact subset $K \subset \rdd$ it is enough to choose $\Psi \in \cS(\rdd)$ with $\Psi \equiv 1$ on $K$. 
\end{enumerate}		
\end{proof}

\section*{Acknowledgements} 
\noindent It is a pleasure to thank Professor Fabio Nicola for helpful comments on various drafts of this note. 

%\bibliographystyle{amsalpha}
%\bibliography{reftesi}

\providecommand{\bysame}{\leavevmode\hbox to3em{\hrulefill}\thinspace}
\providecommand{\MR}{\relax\ifhmode\unskip\space\fi MR }
% \MRhref is called by the amsart/book/proc definition of \MR.
\providecommand{\MRhref}[2]{%
	\href{http://www.ams.org/mathscinet-getitem?mr=#1}{#2}
}
\providecommand{\href}[2]{#2}

\end{document}